\documentclass[final]{elsarticle}

\usepackage{amsthm}
\usepackage{graphicx}
\usepackage[utf8]{inputenc}
\usepackage{amsmath}
\usepackage{amssymb}
\usepackage{color}
\usepackage{subfigure}
\newcommand{\NP}{\ensuremath{\mathcal{NP}}}

\def\O(#1){\ensuremath{\mathcal{O}(#1)}}

\journal{noname}

\newtheorem{theorem}{Theorem}
\newtheorem{lemma}{Lemma}[section]
\newtheorem{corollary}{Corollary}
\newtheorem{definition}{Definition}

\def\squareforqed{\hbox{\rlap{$\sqcap$}$\sqcup$}}
\def\qed{\ifmmode\squareforqed\else{\unskip\nobreak\hfil
\penalty50\hskip1em\null\nobreak\hfil\squareforqed
\parfillskip=0pt\finalhyphendemerits=0\endgraf}\fi}

\makeatletter
\newif\if@restonecol
\makeatother

\usepackage[ruled,vlined,linesnumbered]{algorithm2e}
\DontPrintSemicolon
\definecolor{commentgreen}{rgb}{0, 0.5, 0}
\newcommand{\algcom}[2][0.7]{%
  \hfill\makebox[#1\linewidth][l]{\color{commentgreen}\textit{// #2}}%
}

\begin{document}

\begin{frontmatter}

\title{T-Shape Visibility Representations of 1-Planar Graphs\tnoteref{thanks}}

\author{Franz J.\ Brandenburg}
\ead{brandenb@informatik.uni-passau.de}
\address{University of Passau, 94030 Passau, Germany.}
\tnotetext[thanks]{Supported in part by the Deutsche
Forschungsgemeinschaft
  (DFG), grant Br835/}

\begin{abstract}

A shape visibility representation displays a graph so that each
vertex is represented by an orthogonal polygon  of a particular
shape and for each edge there is a horizontal or vertical line of
sight between the polygons assigned to its endvertices.  Special
shapes are rectangles, \textsf{L}, \textsf{T}, \textsf{E} and
\textsf{H}-shapes, and caterpillars. A flat rectangle is a
horizontal bar of height $\epsilon>0$. A graph is 1-planar if there
is a drawing in the plane such that each edge is crossed at most
once and is IC-planar if in addition no two crossing edges share a
vertex.

We show that every IC-planar graph has a flat rectangle visibility
representation  and that every 1-planar graph has a \textsf{T}-shape
visibility representation. The representations use quadratic area
and can be computed in linear time from a given embedding.

\end{abstract}


\end{frontmatter}

\section{Introduction}

A graph is commonly visualized by a drawing in the plane or on
another surface. In return, properties of   drawings are used to
define properties of graphs. Planar graphs are the most prominent
example. Also, the genus of a graph and $k$-planar graphs are
defined in this way, where a graph is $k$-planar for some $k \geq 0$
if there is a drawing in the plane such that each edge is crossed at
most $k$ times.

Planar graphs admit a different visualization by  bar visibility
representations. A \emph{bar visibility representation} consists of
a set of non-intersecting horizontal line segments, called bars, and
vertical lines of sight between the bars. We assume that the lines
of sight have width $\epsilon > 0$ and also that the bars have
height at least $\epsilon$. Each bar represents a vertex of a graph
and there is an edge if (or if and only if) there is a line of sight
between the bars of the endvertices. Hence, there is a bijection
between vertices and bars and a correspondence between edges and
lines of sight that is one-to-one in the \emph{weak} or
``if''-version and also onto in the \emph{strong} or ``if and only
if''-version. A graph is a \emph{bar visibility graph} if it admits
a bar visibility representation. Other graph classes are defined
analogously.

Bar visibility representations and graphs were intensively studied
in the 1980s and the representations of planar graphs were
discovered independently multiple times \cite{dhlm-rpg-83,
ow-grild-78,rt-rplbopg-86,tt-vrpg-86,w-cbg-85}. Note that strong
visibility with  lines of sight of width zero excludes $K_{2,3}$ and
some 3-connected planar graphs \cite{a-rvg-92} and implies an
NP-hard recognition problem \cite{a-rvg-92}. Obviously, every weak
visibility graph is an induced subgraph of a strong visibility graph
with lines of sight of width zero or $\epsilon >0$.

In the late 1990s  visibility representations were generalized to
represent non-planar graphs. The approach by Dean et
al.~\cite{DEGLST-bkvg-07} admits semi-transparent bars and lines of
sight that  traverse up to $k$ other bars. In other words, an edge
can cross up to $k$ vertices. Some facts are known about   bar
$k$-visibility graphs: for $k=1$ each graph  of size $n$ has at most
$6n-20$ edges  and the bound can be achieved for all $n \geq 8$
\cite{DEGLST-bkvg-07}. In consequence, $K_8$ is the largest complete
bar $1$-visibility graph.  A graph has thickness $k$ if it can be
decomposed into $k$ planar graphs. However, bar 1-visibility graphs
are incomparable to thickness two (or biplanar) graphs, since there
are thickness two graphs with $6n-12$ edges which cannot be bar
1-visibility graphs and conversely   there are bar 1-visibility
graphs with thickness three \cite{fm-pbkvg-08}. Bar 1-visibility
graphs have an NP-hard \cite{bhkn-bvg-15} recognition problem. Last
but not least, every 1-planar graph has a bar 1-visibility
representation which uses only quadratic area and can be specialized
so that a line of sight crosses at most one bar and each bar is
crossed at most once \cite{b-vr1pg-14}. The inclusion relation
between 1-planar and bar 1-visibility graphs was obtained
independently by Evans et al.~\cite{ekllmw-b1vg-14}

Rectangle visibility representations of graphs were introduced by
Hutchinson et al.~\cite{hsv-rstg-99}. Here, each vertex  is
represented by an axis-aligned rectangle and there are horizontal
and vertical lines of sight for the edges, which cannot penetrate
rectangles. Hutchinson et al. studied the strong  version of
visibility. They proved a density of $6n-20$ which is tight for all
$n \geq 8$. In consequence, $K_8$ is the largest rectangle
visibility graph. Rectangle visibility graphs have thickness two
whereas it is unknown whether they have geometric thickness two
\cite{hsv-rstg-99}, which requires a decomposition into two
straight-line planar graphs. The recognition problem for weak
rectangle visibility graphs is \NP-hard \cite{s-RVG-NP-96}.

We generalize rectangle visibility representations  to
$\sigma$-\emph{shape visibility representations}. A \emph{shape}
$\sigma$ is an orthogonal drawing of a ternary tree $\tau$, which is
expanded to an \emph{orthogonal polygon} in a $\sigma$-shape
visibility representation. Thereby, each   edge of $\tau$ is
expanded to a rectangle of  width  $w>0$ and  height $h>0$. The
images of the
 vertices are similar and differ only in the length and width of the
 horizontal and vertical  pieces of the polygon. In particular, rectangle
visibility is \textsf{I}-shape or ``--''-shape visibility. Since
visibility representations can be reflected or rotated by multiples
of 90 degrees we treat the respective shapes as equivalent and shall
identify them. For example, any single element of the set
$\{\lfloor,\rfloor, \lceil, \rceil\}$ can be used for an
\textsf{L}-shape. However, a set of shapes must be used if the
vertices shall have different shapes,
 e.g., $\{ \lfloor,\rfloor, \lceil, \rceil\}$ for  \textsf{L}-shapes in \cite{lm-Lvis-16}.
Other common shapes are    \textsf{H}, \textsf{F} or \textsf{E}.
 A \emph{rake} is a generalized \textsf{E} with
many teeth that are directed upwards, and a \emph{caterpillar} is a
two-sided rake with a horizontal path and vertical lines from the
path to the leaves above and below. The number of teeth or vertical
lines is reflected by the vertex complexity of ortho-polygon
visibility representation \cite{ddelmmw-ovreg-16}.

 In a \emph{flat rectangle visibility representation}
%, called a flat orthogonal box drawing in
%\cite{b-hptpgd-14},
the rectangles have height $\epsilon >0$ where $\epsilon$ is the
width of a sight of line  \cite{tt-vrpg-86}. Then the vertices are
represented by bars, as in bar visibility representations, such that
two bars at the same level can see one another by a horizontal  line
of sight if there is no third bar in between. Moreover, a horizontal
and vertical  line of sight may cross, which is not allowed in the
flat visibility representations by Biedl \cite{b-sdogs-11}.

Shape visibility representations  have been introduced by   Di
Giacomo et al.~\cite{ddelmmw-ovreg-16}. They use caterpillars as
shapes in their results. \textsf{L}-visibility representations have
been introduced by Evans et al.~\cite{elm-svrL-15} using any shape
from the set $\{ \lfloor, \rfloor, \lceil, \rceil\}$. This approach
was adopted  by Liotta and Montecchiani \cite{lm-Lvis-16} for the
representation of IC-planar
graphs.\\

In this work, we prove the following:

\begin{theorem} \label{thm:IC-RVG}
Every n-vertex IC-planar graph $G$ admits a flat rectangle
visibility representation in $O(n^2)$ area, which can be computed in
linear time from a given IC-planar embedding of $G$.
\end{theorem}

\begin{theorem} \label{thm:1-cat}
Every n-vertex 1-planar graph $G$ admits a  \textsf{T}-shape
visibility representation in $O(n^2)$ area, which can be computed in
linear time from a given 1-planar embedding of $G$.
\end{theorem}

The first theorem improves upon a result by Liotta and Montecchiani
\cite{lm-Lvis-16} who use the set $\{ \lfloor, \rfloor, \lceil,
\rceil\}$ as \textsf{L}-shapes. Our result is also a variation of
the bar 1-visibility representation of 1-planar graphs by
Brandenburg \cite{b-vr1pg-14} such that an edge-bar crossing is
substituted by a crossing of a vertical and a horizontal line of
sight.

The second theorem extends a recent result  by Di Giacomo et
al.~\cite{ddelmmw-ovreg-16} and contrasts a result by Biedl et
al.~\cite{blm-vrnpg-16}. However, there are different settings. We
operate in the \emph{variable embedding setting} and admit changing
the embedding. In the other works an \emph{embedding-preserving
setting} is used which enforces a coincidence of the embedding of a
1-planar drawing and a visibility representation.
For the first theorem, we reroute an edge in each $B$-configuration,
as depicted in Figs.~\ref{fig:kite} (b) and (c). The change of the
embedding can be undone with a little effort. However, the full
power of horizontal and vertical lines of sight is used for the
second theorem. Here some crossing edges undergo a separate
treatment and substantially change the embedding. In contrast,
 Di Giacomo et al. have shown that every 1-planar graph
admits a caterpillar visibility representation  and that there are
2-connected 1-planar graphs $G_n$ that need rakes of arbitrary size
if the embedding is preserved.  Biedl et al.~\cite{blm-vrnpg-16}
proved that there is no rectangle visibility representation of $K_6$
that preserves a given 1-planar embedding. However, the graphs $G_n$
and $K_6$ have a rectangle visibility representation (since $K_6$
and the components of $G_n$ are subgraphs of $K_8$).

The paper is organized as follows: In Sect. \ref{sect:prelim} we
recall basic notions and facts on 1-planar graphs and we consider
vertex numberings of planar graphs. We proof Theorem
\ref{thm:IC-RVG} in Sect. \ref{sect:proof1} and Theorem
\ref{thm:1-cat} in Sect.~\ref{sect:proof2} and conclude with general
properties of shape visibility graphs in Sect.~\ref{sect:general}.

\section{Preliminaries} \label{sect:prelim}

We consider simple undirected graphs $G = (V, E)$ with
 a finite set of
vertices $V$ of size $n$ and a finite set of undirected edges $E$.
 It is assumed that the  graphs is 2-connected, since components can be
treated separately or they can be connected by further planar edges.
A \emph{drawing} maps the vertices of a graph to distinct points in
the plane and each edge is mapped to a Jordan arc between the
endpoints. Our drawings are   simple  so that two edges have at most
one point in common, which is either a common endvertex or a
crossing point. A drawing is planar  if edges do not cross and
1-planar  if each edge is crossed at most once. Moreover, in an
IC-planar drawing  each vertex is incident to at most one crossing
edge. A graph is called \emph{planar} (\emph{1-planar, IC-planar})
if it admits a respective drawing. A planar drawing partitions the
plane into topologically connected regions, called faces, whose
boundary consists of edges and edge segments and is specified by a
cyclic sequence of  vertices and crossing points. The unbounded
region is called the outer face. An \emph{embedding}
$\mathcal{E}(G)$ of a graph $G$ is an equivalence class of drawings
of $G$ with the same set of faces. For an algorithmic treatment, we
use the embedding of a planarization   of $G$ which is obtained by
treating the crossing points as  dummy vertices of degree four. An
embedded planar graph  is specified by a  rotation system, which is
the cyclic list of all neighbors or incident edges at each vertex in
clockwise order.

1-planar graphs are the most important class of so-called
beyond-planar graphs. Beyond-planarity comprises graph classes that
extend the planar graphs and are defined by specific restrictions of
crossings. 1-planar graphs were studied first by Ringel
\cite{ringel-65} who showed that they are at most 7-colorable. In
fact, 1-planar graphs are 6-colorable \cite{b-np6ct-95}. Bodendiek
et al.~\cite{bsw-bs-83, bsw-1og-84} observed that 1-planar graphs of
size $n$ have at most $4n-8$ edges and that this bound is tight for
$n=8$ and all $n \geq 10$. This fact was discovered independently in
many works. In consequence, an embedding has linear size and can be
treated   in linear time. \emph{IC-planar} (independent crossing
planar) graphs  are an important special case \cite{a-cnircn-08}. An
IC-planar graph has at most $3.25n-6$ edges \cite{ks-cpgIC-10} and
the bound is tight. In between are NIC-planar graphs
\cite{zl-spgic-13} which are defined by 1-planar drawings in which
two pairs of crossing edges share at most one vertex. Their density
is at most $3.6(n-2)$. 1-planar, NIC-planar, and IC-planar graphs
have some properties in common: First, there is a difference between
densest and sparsest graphs. A sparsest graph cannot be augmented by
another edge and has as few edges as possible whereas a densest
graph has as many edges as possible. It is known that there are
sparse 1-planar graphs with $\frac{45}{17}n- \frac{84}{17}$ edges
\cite{begghr-odm1p-13}, sparse NIC-planar graphs with $3.2(n-2)$
\cite{bbhnr-NIC-16} edges and sparse IC-planar graphs with $3n-4$
edges \cite{bbhnr-NIC-16}. The NP-hardness of the recognition
problems was discovered independently multiple times
\cite{GB-AGEFCE-07, km-mo1ih-13, abgr-1prs-15,bbhnr-NIC-16,
bdeklm-IC-16} and holds even if the graphs are 3-connected and are
given with a rotation system.
%that defines the cyclic ordering of the incident edges at each
%vertex.
On the other hand, triangulated graphs can be recognized in cubic
time \cite{cgp-rh4mg-06,b-4m1pg-15}. A \emph{triangulated graph}
admits a drawing so that all faces are triangles. Then all pairs of
crossing edges induce $K_4$ as a subgraph.

The most remarkable distinction between IC-planar and NIC-planar
graphs is their relationship to RAC graphs. A graph is RAC (right
angle crossing) \cite{del-dgrac-11} if it admits a straight-line
drawing such that edges cross at a right angle. RAC graphs have at
most $4n-10$ edges, and if they meet the upper bound, then they are
1-planar \cite{el-racg1p-13}. In contrast, there are 1-planar graphs
that are not RAC and RAC graphs that are not 1-planar
\cite{del-dgrac-11}. Hence, 1-planar graphs and RAC graphs are
incomparable. Recently, Brandenburg et al.~\cite{bdeklm-IC-16}
showed that every IC-planar graph is a RAC graph and Bachmaier et
al.~\cite{bbhnr-NIC-16} proved that RAC graphs and NIC-planar graphs
are incomparable.

\section{Planar and 1-Planar Graphs}

For our algorithms we use two tools: triangulated 1-planar
embeddings and an st-numbering. We need the following versions of a
given 1-planar graph $G$:  $G_{\boxtimes}, G_{\boxplus},
G_{\square}$ and $G_{\bullet}$. Each version is obtained  from an
embedding $\mathcal{E}(G)$ and   inherits the embedding. Graphs
$G_{\boxtimes}$ and $G_{\boxplus}$ are supergraphs of $G$ which
coincide on 3-connected graphs, $G_{\boxplus}$, $G_{\square}$ and
$G_{\bullet}$ admit multi-edges, and $G_{\square}$ and $G_{\bullet}$
are planar.

First, augment the embedding $\mathcal{E}(G)$  by as many planar
edges as possible and thereby obtain a  planar maximal embedding
  $\mathcal{E}(G_{\boxtimes})$ of $G_{\boxtimes}$ \cite{abk-sld3c-13}. Then the
endvertices of each pair of crossing edges induce $K_4$. Each such
$K_4$ should be embedded as a kite with the crossing point inside
the boundary of the 4-cycle of the endvertices and no other vertex
inside this boundary, see Fig.~\ref{fig:kite}(a). Otherwise, there
are B- or W-\emph{configurations} \cite{t-rdg-88}, as shown in
Figs.~\ref{fig:kite} (b) and (d) or there is a separation pair as in
Fig.~\ref{fig:kite}(e), where the inner components are contracted to
a single vertex. B-configurations can be removed by changing the
embedding. Therefore, choose the other face next to the edge
$\{a,b\}$ between the vertices of a separation pair as outer face
and reroute   $\{a,b\}$, as illustrated in Fig.~\ref{fig:kite}(c),
or  flip the component. Thereafter we  add further planar edges if
possible.  For example, in Fig.~\ref{fig:kite}(c) one may connect
$x$ with another vertex by a planar edge. Then at most one
W-configurations remains in the outer face if the graphs are
3-connected \cite{abk-sld3c-13}. If the graph is 3-connected, then
we take the obtained embedding as a \emph{normal form}
\cite{abk-sld3c-13}. It corresponds to a triangulation of the
planarization with crossing points as vertices of degree four.
Otherwise, there are separation pairs and pairs of crossing edges
that separate the components, as sketched in Fig.~\ref{fig:separate}
and shown in Fig.~\ref{fig:XWdouble}.

\begin{figure}
  \centering
  \subfigure[ ]{
    \parbox[b]{1.8cm}{%
      \centering
      \includegraphics[scale=0.45]{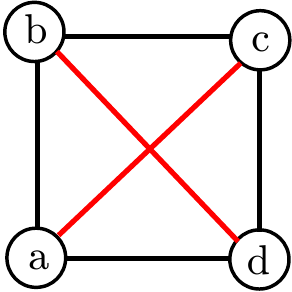}
    }
    \label{fig:kite1}
  }
  \hfil
  \subfigure[ ]{
    \parbox[b]{1.8cm}{%
      \centering
      \includegraphics[scale=0.45]{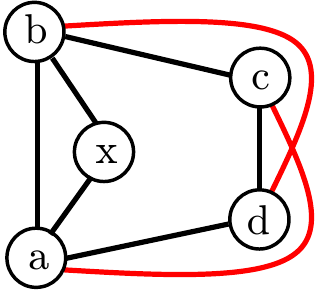}
    }
    \label{fig:B-config}
  }
  \hfil
    \subfigure[ ]{
    \parbox[b]{1.8cm}{%
      \centering
      \includegraphics[scale=0.45]{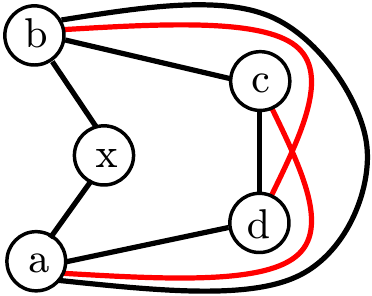}
    }
    \label{fig:W-config}
  }
  \hfil
    \subfigure[ ]{
    \parbox[b]{1.8cm}{%
      \centering
      \includegraphics[scale=0.45]{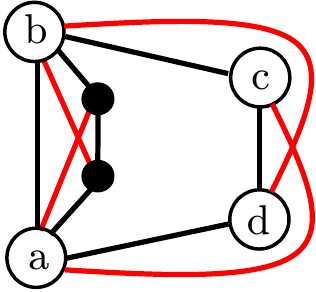}
    }
    \label{fig:W-config}
  }
  \hfil
    \subfigure[ ]{
    \parbox[b]{2.0cm}{%
      \centering
      \includegraphics[scale=0.45]{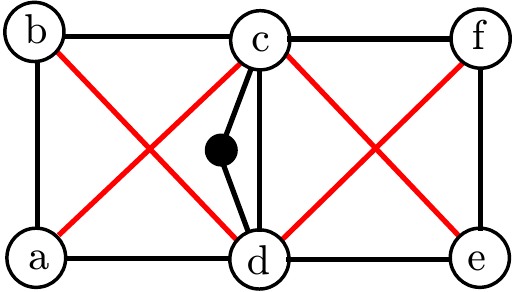}
    }
    \label{fig:W-config}
  }
  \caption{(a) a kite, (b) a B-configuration, (c) a rerouted B-configuration,   (d) a W-configuration, and (e)
  a separation pair with an inner component represented by a dot.}
  \label{fig:kite}
\end{figure}

Graph $G_{\boxplus}$ extends $G_{\boxtimes}$ by multi-edges at
separation pairs and the removal of B-configurations.
 Let $[x,y]$ be a separation pair so that
$G_{\boxtimes} - \{x,y\}$ decomposes into components $H_0,
H_1,\ldots, H_p$ for some $p \geq 1$. By recursion there is a
decomposition tree,  which is a simplified version of the
SPQR-decomposition tree \cite{dt-olpt-96, gm-ltist-01} and can be
computed in linear time. Let $H_0$ be the \emph{outer component} and
let $H_1, \ldots, H_p$ be   \emph{inner components} which are
children of the outer component in the decomposition tree. Expand
each inner component $H_i$ to $\widehat{H_i}$ which includes $x$ and
$y$ and the edges between $x$ and $y$ and vertices of $H_i$. Flip
and permute the inner components in the embedding
$\mathcal{E}(G_{\boxtimes})$ and add further planar edges  so that
no B-configuration remains. An expanded inner component
$\widehat{H_i}$ is embedded as a W-configuration and $\widehat{H_i}$
and $\widehat{H_j}$
 for $i \neq j$ are separated by two pairs of crossing edges.
  Now, an embedding of
$G_{\boxplus}$ is obtained by  adding a copy $e_i$ of the edge $e_0
= \{x,y\}$ between $\widehat{H_i}$  and $\widehat{H_{i+1}}$ and
beyond $H_q$  for $i=1,\ldots, p$ with $H_{p+1} = H_0$
\cite{b-vr1pg-14}, as illustrated Fig.~\ref{fig:separate}. Note that
$G_{\boxplus} = G_{\boxtimes}$ if $G$ is 3-connected.
Graph $G_{\square}$ is obtained from
 of $G_{\boxplus}$ by removing all pairs
of crossing edges in an embedding  $\mathcal{E}(G_{\boxplus})$. Due
to the multi-edges,   the embedding of $G_{\square}$ has triangles
and quadrangles with a quadrangle for each pair of crossing edges.
 Finally,
$G_{\bullet}$ is obtained from an embedding $G_{\boxtimes}$ of an
IC-planar graph $G$ by the contraction of each kite to a single
vertex.

\begin{figure}
   \begin{center}
     \includegraphics[scale=0.7]{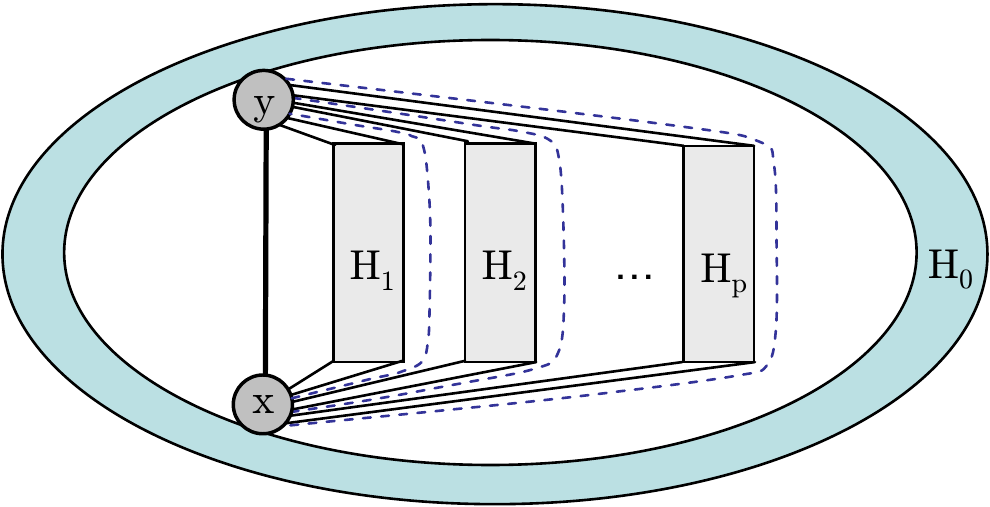}
     \caption{A separation pair $[x,y]$ and a separation of the inner components.}
     \label{fig:separate}
   \end{center}
\end{figure}

\begin{lemma} \label{lem:normalform}
Let $\mathcal{E}(G)$ be a 1-planar embedding of a 1-planar graph
$G$.
\begin{itemize}
  \item The embedding $\mathcal{E}(G_{\boxplus})$ is triangulated.
  \item A pair of crossing edges is embedded as a kite or there is
  a \emph{W}-configuration and a separation pair.
  \item $G_{\boxplus}$  has at most 4n-8 edges.
  \item $\mathcal{E}(G_{\boxplus})$  and $\mathcal{E}(G_{\boxplus}$ can be computed in linear time
  from $\mathcal{E}(G)$.
\end{itemize}
\end{lemma}

\begin{proof}
The planar maximal embedding of each 3-connected component is
triangulated \cite{abk-sld3c-13}. It may change the embedding by
rerouting an edge of a B-configuration, which thereby is turned into
a kite, see Figs.~\ref{fig:kite} (b) and(c). Then the stated
properties hold for $\mathcal{E}(G_{\boxtimes})$. At each separation
pair $[x,y]$, the multi-edge  between two components induces a
triangulation with triangles consisting of $x, y$ and a crossing
points of two edges incident to $x$ and $y$. After an elimination of
all B-configurations there is a kite or a W-configuration for each
pair of crossing edges.

Concerning the number of edges, at every separation pair $[x,y]$
with inner components $H_1,\ldots,H_p$ and a 4-cycle $(a_i, b_i,
c_i, d_i)$ as outer boundary of $H_i$ replace the pairs of crossing
edges $\{x, c_i\}, \{y, d_i\}$ and
 $\{x, a_{i+1}\}, \{y,
b_{i+1}\}$ by a kite with edges $\{c_i, a_{i+1}\},  \{c_i,
b_{i+1}\}, \{d_i, a_{i+1}\}, \{d_i, b_{i+1}\}$ and replace the i-th
copy of $\{x, y\}$ between $H_i$ and $H_{i+1}$ by   the edges
$\{a_i, a_{i+1}\}$ and $\{b_i, b_{i+1}\}$ for $i=1,\ldots, p-1$. The
resulting graph is 1-planar and has $q-1$ more edges than $G$.
 Hence, there are
at most $4n-8$ edges. Each step from $\mathcal{E}(G)$ to
$\mathcal{E}(G_{\boxplus})$ takes linear time and is performed on
the embedding of the planarization.
\end{proof}

Concerning the density of 1-planar graphs, each  W-configuration
reduces the maximum number of edges by two, since each pair of edges
crossing in the outer face can be substituted by four edges. This
parallels the situation of planar graphs and 2-connected
components.\\

Next, we consider vertex orderings of planar graphs which are later
applied to graphs $G_{\square}$ and $G_{\bullet}$.

Let $\{s,t\}$ be an edge of a planar graph in the outer face of an
embedding of $G$. An \emph{st-numbering} is an ordering $v_1,
\ldots, v_n$ of the vertices of $G$ such that $s=v_1$, $t=v_n$ and
every vertex $v_i$ other than $s$ and $t$ is adjacent to at least
two vertices $v_j$ and $v_k$ with $j<i<k$. It is known that every
2-connected graph has st-numberings and an st-numbering can be
constructed in linear time \cite{et-stnum-76} for every edge
$\{s,t\}$. An st-numbering induces an orientation of the edges of
$G$ from a low ordered vertex to a high ordered one, called a
\emph{bipolar orientation}.\\

For convenience, we identify each vertex with its st-number and with
its orthogonal polygon in a shape visibility representation and
consider each edge as oriented. In simple words a vertex $u$ is less
than vertex $v$ and vertex $v$ is placed at some point.\\

 If $G$ is planar, then a bipolar orientation
transforms $G$ into an upward planar graph and   partitions the set
of edges incident to a vertex $v$  into a sequence of incoming and a
sequence of outgoing edges \cite{dett-gdavg-99}. Accordingly, each
vertex has two lists of faces below and above it, which are ordered
clockwise or left to right. A face is below $v$ if both edges
incident to $v$ are incoming edges, and above it, otherwise. In
addition, at each separation pair $[x,y]$ with components $H_0,
\ldots, H_k$, the vertices of each inner component $H_i$ with $i
\geq 1$ are ordered consecutively and they appear between $x$ and
$y$ if  $x<y$. We can write $x < H_1 < \ldots < H_k < y$, where
$H_1,\ldots,H_k$ is any permutation of the inner components. For
example, one may choose the cyclic ordering at $x$ if an embedding
is given. st-numberings are a useful tool for the construction of
visibility representations of planar graphs \cite{dett-gdavg-99}.

\emph{Canonical orderings} are used for straight-line drawings of
planar graphs. They were introduced by de Fraysseix et al.
\cite{fpp-hdpgg-90} for triangulated planar graphs and were
generalized to 3-connected \cite{k-dpguco-96} and to 2-connected
graphs \cite{hs-asdpg-98}. The subsequent definition is taken from
\cite{bbc-leftish-11}.

\begin{definition}
Let $ \Pi = (P_0,\ldots, P_q)$ be a partition of the set of vertices
of a graph $G$ of size $n\geq 5$ into paths such that $P_0 = \langle
v_1, v_2 \rangle$, $P_q = \langle v_n \rangle$   and $\langle  v_1,
P_q, v_2 \rangle$ is the outer face in clockwise order. For $k=0,
\ldots, q$ let $G_k$ be the subgraph induced by $V_k = P_0 \cup
\ldots \cup P_k$ and let $C_k$ be the outer face of $G_k$, called
contour. Then $\Pi$ is a \emph{canonical ordering} if for each
$k=1,\ldots, q-1$:

\begin{enumerate}
  \item $C_k$ is a simple cycle.
  \item Each vertex  $z_i$ in $P_k$ has a neighbor in $V-V_k$.
  \item $|P_k|=1$ or each vertex $z_i$ in $P_k$ has exactly two
  neighbors in $G_k$.
\end{enumerate}
\end{definition}

A canonical ordering $\Pi$ is refined into a vertex ordering $v_1,
\ldots, v_n$ by ordering the vertices in each $P_k$, $k>0$, straight
or in reverse.

A canonical ordering can be computed by a peeling technique which
successively removes the vertices of the paths in reverse order
starting from $P_q$. For a quadrangle it would consist  of two paths
of length two.  Care must be taken that the removal of the next path
$P_k$ preserves the 2-connectivity of $G_i$   for $i=1,\ldots, k$,
see \cite{bbc-leftish-11,k-dpguco-96}. The contour $C_k$ is ordered
left to right with $v_1$ at the left and $v_2$ at the right so that
edge $\{v_2, v_1\}$ closes the cycle.

A path $P$ is a \emph{feasible candidate}  for step $k+1$ of $\Pi =
(P_0,\ldots, P_q)$ if also $(P_0,\ldots, P_k,P)$ can be extended to
a canonical ordering of $G$.

\begin{definition}
A canonical ordering $\Pi = (P_0,\ldots,P_q)$  is called
\emph{leftish} if for $k=0,\ldots, q-1$ the following is true: Let
$c_{\ell}$ be the left neighbor of $P_{k+1}$ on $C_k$  and let $P$
be a feasible candidate for step $k+1$ with left neighbor
$c_{\ell'}$. Then $c_{\ell} < c_{\ell'}$.
\end{definition}

A leftish canonical ordering of a 3-connected planar graph can be
computed in linear time \cite{bbc-leftish-11}. For 2-connected
planar graphs we extend the ordering as in the st-numbering case. At
each separation pair $[x,y]$ with $x<y$ remove the inner components
and compute the leftish canonical ordering of the 3-connected
remainder. Then compute the leftish canonical ordering of each
component and insert them just before $y$.\\

We study some properties of    leftish canonical orderings on upward
planar  graphs. The orientation of the edges and upward direction is
obtained from the (extended) leftish canonical ordering, which is an
st-numbering with $s=1$ and $t=n$. Each edge has a face to its left
and to its right if the graphs are 2-connected. Each face $f$ has a
source and a sink, called $bottom(f)$ and $top(f)$, respectively.
Suppose that edge $\{s,t\}$ is routed at the left. We call the face
to the right of $\{s,t\}$ the leftmost face and the outer face is
called the
rightmost face \cite{dett-gdavg-99}.\\

In the remainder of this section, let $G$ be a 3-connected planar
graph $G$ with an st-numbering whose faces are triangles or
quadrangles that are traversed clockwise. The outer face is
excluded.

\begin{definition}
A quadrangle $f = (a,b,c,d)$ is called a \emph{rhomboid} with bottom
$a$, left end $b$, right end $d$, and top $c$ if there are two paths
$\langle a,b,c \rangle$ and $\langle a,d,c \rangle$ enclosing $f$
with $b$ to the left of $f$ and $d$ to the right. Face $f$ is a
\emph{left-trapezoid} if there is an edge $\{a,d\}$ to the right of
$f$ and a path $\langle a, b, c,d \rangle$ to the left. A
\emph{right-trapezoid} is defined accordingly, see
Fig.~\ref{fig:rhombus}.
\end{definition}

%These types of faces have been addressed in \cite{b-vr1pg-14,
%help-ft1pg-12}.\\

\begin{figure}
  \centering
  \subfigure[ ]{
    \parbox[b]{2.1cm}{%
      \centering
      \includegraphics[scale=0.6]{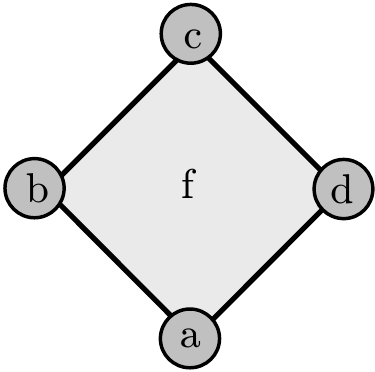}
    }
    \label{fig:kite1}
  }
  \hfil
  \subfigure[ ]{
    \parbox[b]{2.1cm}{%
      \centering
      \includegraphics[scale=0.6]{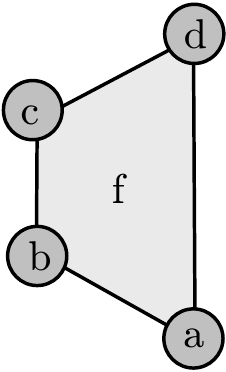}
    }
    \label{fig:B-config}
  }
  \hfil
    \subfigure[ ]{
    \parbox[b]{2.1cm}{%
      \centering
      \includegraphics[scale=0.6]{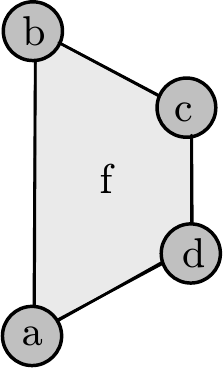}
    }
    \label{fig:W-config}
  }
  \caption{(a) a rhomboid, (b) a left-trapezoid, and (c)  a right-trapezoid}
  \label{fig:rhombus}
\end{figure}

First, each path $P_k$ of a leftish canonical ordering $\Pi$ has
length at most two, since one of $v_1$ and $v_r$ has at least two
neighbors on $C_k$ if $P_k = \langle v_1, \ldots, v_r \rangle$ with
$r\geq 3$. Otherwise,  there are faces as $m$-gons with $m >4$. If
$P_k$ has length two, then it is inserted into $C_k$. Otherwise,
$C_{k+1}$ is obtained by replacing a subsequence $\gamma$ of $C_k$
by $v$ with $P_k = \langle v \rangle$, where the vertices in
$\gamma$ are \emph{covered} by $v$ \cite{fpp-hdpgg-90, k-dpguco-96}.

Second, if $P_k = \langle v_1,  v_2 \rangle$ is a path of length
two, then the face  below $P_k$ is a quadrangle $f_k = ( u_1, v_1,
v_2, u_2)$. We prefer rhomboids over trapezoids and therefore direct
$P_k$ from $v_2$ to $v_1$ if $u_2 = bottom(f_k)$ and otherwise from
$v_1$ to $v_2$. If $P_k = \langle v \rangle$ is a singleton then it
may cover several faces, which are triangles, rhomboids, or
trapezoids. We say that face $f_k$ is \emph{covered} by $P_k$.

Third, we consider faces. For a vertex $v$ on a contour $C_k$ let
  $f_1(v),\ldots, f_{\nu}(v)$ be the left to right ordering of
the faces incident to $v$ and above $C_k$ which is defined by the
clockwise ordering of the outgoing edges. The outer face is
discarded. For each quadrangle $f_i(v) = (v, b,c,d)$ let $t_i =
top(f_i(v))$. Then $t_i=b$ if $f_i(v)$ is a right-trapezoid, $t_i=c$
if $f_i(v)$ is a rhomboid and $t_i=d$ if $f_i$ is a left-trapezoid
and $t_i$ covers $f_i$. A face $f_i(v) = (v,b,d)$ or $f_i(v) =
(v,b,c,d)$ for $i=1,\ldots, \nu$ has \emph{left-support} if there is
a contour $C_{\ell} = (1, \ldots, b, v, \ldots, 2)$ and
\emph{right-support} if $C_{\ell} = (1, \ldots, v, d, \ldots, 2)$.
If $f_i(v)$ has left-support and $d$ is in $P_{k+j_i}$ for some $i_j
\geq 1$ then either $d$ is immediately to the left of $v$   on the
contour $C_{k+j_i}$ or the placement of $d$ covers $v$. The case is
symmetric to the right if $f_i(v)$ has right support.

For an illustration see Fig.~\ref{fig:faces}.

\begin{figure}
   \begin{center}
     \includegraphics[scale=0.55]{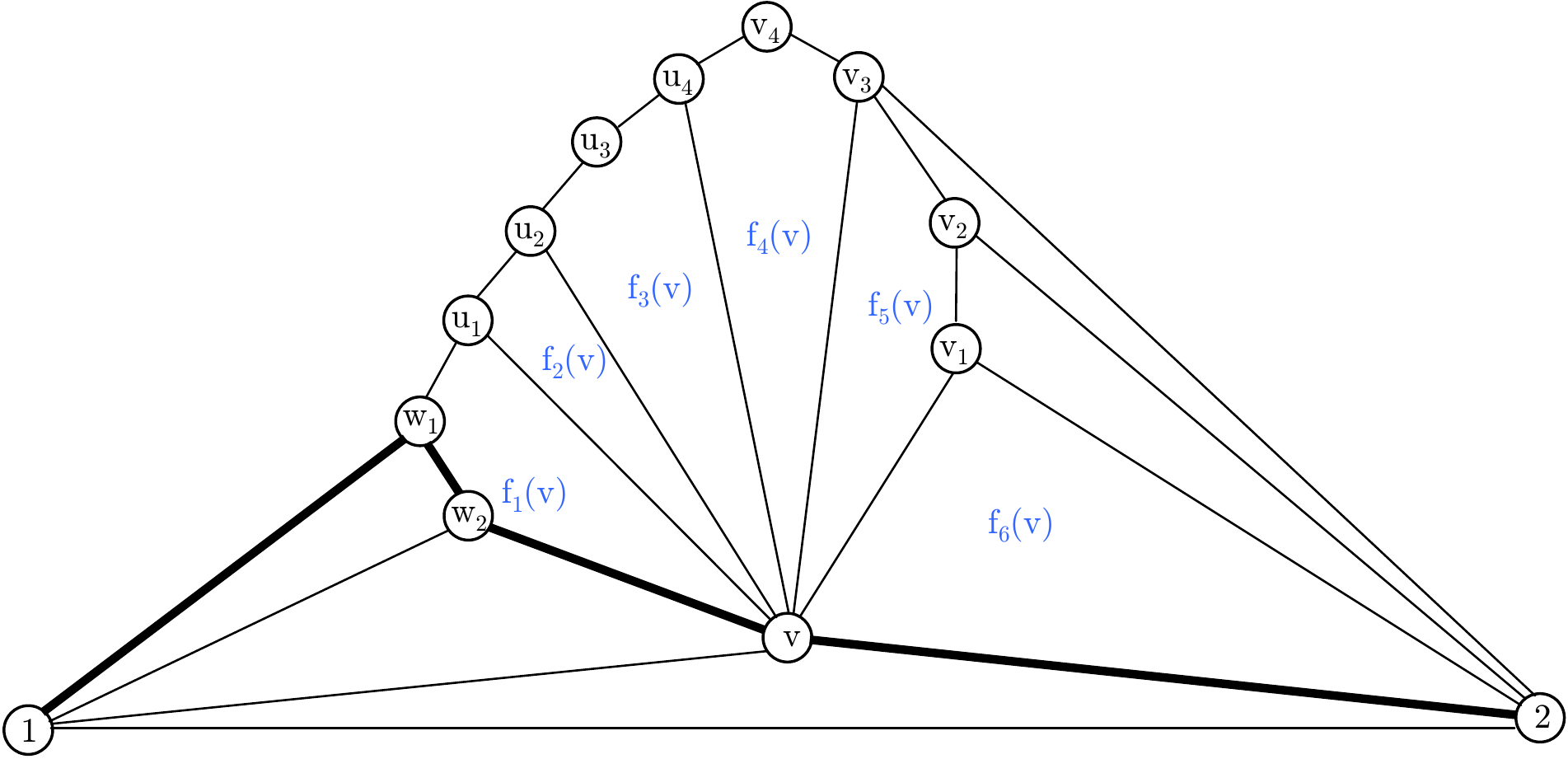}

     \caption{A sequence of faces above vertex $v$. Starting from a contour $C_k = (1, w_1, w_2, v, 2)$
     there is a leftish canonical ordering $\Pi = (\langle u_1
     \rangle, \langle u_2  \rangle,
     \langle u_3, u_4 \rangle,  \langle v_1 \rangle, \langle v_2 \rangle,  \langle v_3 \rangle,  \langle v_4
     \rangle)$.
   Face $f_1(v)$ is a left-trapezoid, $f_2(v)$ and
   $f_6(v)$ are
   triangles, $f_4(v)$ is a rhomboid, and $f_5(v)$ is a
   right-trapezoid.
Face $f_2(v)$ is a rhomboid if edge $\{u_3, u_4\}$ is oriented from
$v_4$ to $v_3$
   and a left-trapezoid, otherwise. Face $f_1(v)$ has left-support,
   $f_5(v)$ has right-support, and $f_4(v)$ has left- and
   right-support. The sequence $(u_1, u_2,u_3, u_4, v_4, v_3,
   v_2, v_1)$ of neighbors of $v$ above $C_k$   is bitonic.}

     \label{fig:faces}
   \end{center}
\end{figure}

\begin{lemma} \label{lem:leftish}
Let $v$ be a vertex on a  contour $C_k$. A quadrangle $f$ with
$v=bottom(f)$
 has left-support (right-support)  if $f$ is a left-trapezoid
(right-trapezoid).
\end{lemma}

\begin{proof}
For a contradiction, suppose that $f = (a,b,c,d)$ is a
left-trapezoid and has no left-support. Since $f$ is a
left-trapezoid, vertex $b$ appears before vertices $c$ and $d$ in
the vertex ordering. If $b$ has a neighbor to the right of $v$ on
the contour, then $f$ cannot be a left-trapezoid.
\end{proof}

Note that the statement may not apply to the outer face, which later
on may need a spacial treatment. The type of quadrangles $f(v)$ is
determined by their support and the length of the path with the top
vertex in the leftish canonical ordering.

Let $f_1(v), \ldots, f_{\nu}(v)$ be the left to right (clockwise)
ordering of faces with bottom $v$ above $C_k$ and let
  $j_1 < \ldots < j_{\mu}$ be the subsequence of quadrangles. The
  other faces are triangles.
  For $i=1, \ldots, \nu$ let $P_{t_i}$ contain the top vertex of
  $f_{i}$ so that $f_{i}$ is closed by $P_{t_i}$.  Thus, $t_i$ is the time stamp for the completion of
  face  $f_{j_i}$ in a canonical ordering and a drawing based on it.

\begin{lemma} \label{lem:decidetrapezoid}
If a face $f=f_i(v)$ for $i=1, \ldots, \nu$  is a triangle, then it
has left-support or right-support and $P_{t_i}$ is a singleton. If
$f$ is a quadrangle, then
 $f$ is a left-trapezoid if
$f$ has  no right-support and $P_{t_i}$ is a singleton.  Face $f$ is
a right-trapezoid if $f$ has  no left-support and $P_{t_i}$ is a
singleton, and $f$ is a rhomboid if $f$ has left and right-support
or $P_{t_i}$  is a path of length two.
\end{lemma}

\begin{proof}
If $f$ is a triangle, then it must be closed by a path of length one
of the leftish canonical ordering and therefore it has left- or
right-support.

Each quadrangle $f$ has a left- or a right-support, since the paths
have length at most two. If $f$ has no right-support and $P_{t_i}$
is a singleton, then vertices $b$ and $c$ are placed before vertex
$d=top(f)$ if $f=(v,b,c,d)$ in clockwise order and $f$ is a
left-trapezoid. The case to the right is symmetric. If $f=(v,b,c,d)$
has left- and right-support, then $c=top(f)$ and $b$ and $d$ are
less than  $c$ in the vertex numbering, so that $f$ is a rhomboid.
If $P_{t_i}$  is a path of length two, then $f$ is a rhomboid by
construction.
\end{proof}

The leftish canonical ordering also determines the order in which
the   vertices of the faces at $v$ are placed.

\begin{lemma} \label{lem:bitonic}
For a leftish canonical ordering $\Pi$ and a vertex $v$ on a contour
$C_k$, the  clockwise sequence of neighbors  $w_1, \ldots, w_z$ of
$v$ above $C_k$ is bitonic, i.e., there is some $m$ with $1 \leq m
\leq z$ such that $w_1 < \ldots < w_m$, $w_{m+1} > \ldots > w_z$ and
$w_m < w_z$ in $\Pi$.
\end{lemma}

For an illustration, see Fig.~\ref{fig:faces} and observe that the
  sequence of the vertices corresponds to the order in which
the faces above $v$ are completed, first from left to right   and
then from right to left.

\begin{proof}
Consider the clockwise sequence of faces $f_1(v),\ldots,
f_{\eta}(v)$ with $\eta \leq z$. Then there is some $\mu$ so that
$f_j$ has left-support for $i=1,\ldots, \mu$ and $f_j$ has
right-support for $\mu+1,\ldots, \eta$. Otherwise, suppose for some
$\kappa$ with $1 \leq \kappa < \mu-1$  face $f_{\kappa}$ has
right-support and face $f_{\kappa+1}$ has left-support. Then
$C_{j_{\kappa}} = (1,\ldots, v, d, \ldots, 2)$ for some vertex $d$
of $f_{\kappa}$ and $C_{j_{\kappa+1}} = (1,\ldots, d, v, \ldots,
2)$, a contradiction.

In consequence, for $1 \leq j < \mu$ and vertices $x$ in $f_j(v)$
and $y$ in $f_{j+1}(v)$ it holds that $x < y$ in $\Pi$. Similarly,
we have $x > y$ for vertices $x$ in $f_j(v)$ and $y$ in $f_{j+1}(v)$
and $\mu+1 \leq j \leq z$. Paths of length two of $\Pi$ are ordered
in accordance with this ordering. Now let $w_1, \ldots, w_m$ be the
vertices in faces $f_1(v),\ldots, f_{\mu}(v)$. By planarity, the
ordering of the faces is in accordance with the ordering of the
vertices in $\Pi$, which is bitonic.
\end{proof}

Note that the sequence of neighbors of $v$ is not continuous in the
leftish canonical ordering. In general, a neighbor $w_i$ of vertex
$v$ on $C_k$ is a right-support of some face $f(u)$ for a vertex $u$
to the left of $v$ on $C_k$.

\section{Rectangle Visibility Representation of IC-planar Graphs} \label{sect:proof1}

The proofs of Theorems \ref{thm:IC-RVG} and \ref{thm:1-cat} are
constructive. The outline of the algorithms is as follows: Take an
embedding $\mathcal{E}(G)$ as a witness for IC-planarity and
1-planarity, respectively. The embedding is first augmented to
$\mathcal{E}(G_{\boxplus})$ as given in Lemma \ref{lem:normalform}.
Since planar maximal IC-planar graphs are 3-connected we can use
$\mathcal{E}(G_{\boxtimes})$ in this case. Thereby, some edges can
be rerouted and some are multiplied to separate components. If the
input were a graph, then constructing the normal form embedding is
an NP-hard problem, since the general recognition problem is NP-hard
and is solvable in polynomial time for graphs with a normal form
embedding \cite{b-4m1pg-15, cgp-rh4mg-06}. Next, $G_{\boxplus}$ is
planarized to $G_{\square}$ by a removal of all pairs of crossing
edges while preserving the (new) embedding. The planarization is
obtained via $G_{\bullet}$  for IC-planar graphs. In a nutshell, the
algorithms use a standard algorithm for  the construction of a
  visibility representation of a planar graph, see \cite{dett-gdavg-99, k-mcvr-97,
  rt-rplbopg-86,tt-vrpg-86}.
  Finally, the pairs of crossing edges are added to the planar
  visibility representation  of $G_{\square}$.
  In case of  IC-planar graphs, there is a quadrangle $f = (a,b,c,d)$  which is drawn as a rhomboid with
 a vertex $b$ to the left and a vertex $d$ to the right of $f$ and at the same
 level ($y$-coordinate)  into which a pair of
 crossing edges is inserted. In the second case, we use a  leftish canonical
 ordering for an st-numbering
 and the capabilities of
 horizontal and vertical lines of sight in the weak visibility
 version.

First, we define   kite-contraction and kite-expansion operations on
IC-planar embeddings in normal form. By IC-planarity, two kites have
no common vertex and do not intersect so that kite-contractions do
not interfere. Moreover, each pair of crossing edges is embedded as
a kite and $G$ is 3-connected  if $\mathcal{E}(G)$ is an IC-planar
embedding in normal form,\cite{abk-sld3c-13, bbhnr-NIC-16}

\begin{definition}
Let $\mathcal{E}(G)$ be an IC-planar embedding in normal form and
suppose there is an $st$-ordering of $G$.

A \emph{kite-contraction} contracts a  kite $\kappa$ with boundary
$(a,b,c,d)$ of $\mathcal{E}(G)$ to a single vertex $v_{\kappa}$ so
that $v_{\kappa}$ inherits all incident edges and henceforth has
multi-edges. A \emph{kite-expansion} is the inverse operation on the
boundary and replaces $v_{\kappa}$ by the 4-cycle $(a,b,c,d)$. Both
operations are adjacency preserving so that a kite-contraction
followed by a kite-expansion just removes the pair of crossing edges
of $\kappa$.

The kite-contraction $G_{\bullet}$ of  $G$ is obtained by
contracting all kites of $\mathcal{E}(G)$.
\end{definition}

\begin{lemma}
Let  $\mathcal{E}(G)$ be an IC-planar embedding in normal form. Then
graph $G_{\bullet}$   is a 3-connected planar graph, which can be
computed from $\mathcal{E}(G)$ in linear time.
\end{lemma}
\begin{proof}
For planarity, first remove one edge from each pair of crossing
edges of each kite, which results in a graph $G_{\triangle}$. Then
$G_{\triangle}$ is a triangulated  planar graph that inherits its
embedding from $\mathcal{E}(G)$. Thus it is 3-connected. Next
contract the edges that remain from each kite to a vertex
$v_{\kappa}$. Since the planar graphs are closed under taking
minors, the edge contractions preserve  planarity and yield
$G_{\bullet}$, since by IC-planarity each vertex $v$ either remains
or is contracted to a vertex $v_{\kappa}$. The removal of
multi-edges results in a   triangulated planar graph, which is
3-connected, and so is $G_{\bullet}$.

It takes linear time   to obtain $G_{\triangle}$ from
$\mathcal{E}(G)$ and $G_{\bullet}$ from $G_{\triangle}$.
\end{proof}

Note that this type of kite-contractions cannot be applied to
1-planar (or NIC-planar) graphs, since vertices may belong to
several kites. Instead one may contract a kite to a single vertex
which
corresponds to its crossing point or consider the $K_4$ network \cite{begghr-odm1p-13}. \\

Next, we consider rhomboidal st-numberings.

\begin{definition}
An embedding $\mathcal{E}(G)$ of a 1-planar graph in normal form is
called \emph{rhomboidal} with respect to an st-numbering if the
$K_4$ subgraph induced by a pair of crossing edges is embedded as a
kite whose boundary is a rhomboid.
\end{definition}

Rhomboidal embeddings distinguish IC-planar graphs from NIC-planar
graphs.

\begin{lemma}
For every  IC-planar graph $G$ and every planar edge $\{s,t\}$ there
is a rhomboidal embedding   which can be computed in linear time
from $\mathcal{E}(G)$.
\end{lemma}

\begin{proof}
First, construct an IC-planar embedding in normal form with
$\{s,t\}$ in the outer face. Next, compute a kite-contraction
$\mathcal{E}(G_{\bullet})$ and an st-numbering of $G_{\bullet}$.
Then do the kite-expansion and extend the st-numbering of
$G_{\bullet}$ to an st-numbering of $G$ as follows: For each
contracted kite $\kappa$  determine a top and a bottom vertex and
then the left and right ends of the face $f$ of $\kappa$ without the
pair of crossing edges.
Expand $\kappa$ in $\mathcal{E}(G_{\bullet})$. If there is exactly
one vertex $u$ of $\kappa$ with only incoming (multi-)edges, then
let $u = bottom(f)$. Choose the top vertex opposite to $u$, and the
left and right ends to the left and right of $f$. Similarly, choose
$v=top(f)$ if only $v$ has outgoing (multi-)edges and choose
$bottom(f)$ opposite to $v$. Otherwise, choose a pair of opposite
vertices so that $u=bottom(f)$ has incoming and $v=top(f)$ has
outgoing (multi-)edges and determine the left and right ends.

Clearly, each step takes linear time.
\end{proof}

We are now able to describe 1-planar graphs that admit a right angle
crossing drawing which is a step towards the intersection of
1-planar and RAC graph that is asked for in \cite{bdeklm-IC-16,
el-racg1p-13}.

\begin{theorem}  \label{thm:rhomboidal}
If $G$ is a 3-connected 1-planar graph so that the augmentation
$G_{\boxtimes}$ has a rhomboidal embedding with respect to a
canonical ordering, then $G$ is a RAC graph.
\end{theorem}
\begin{proof}
Our algorithm is a simplification of the technique used in   Case 1
of the proof of Theorem 2 in \cite{bdeklm-IC-16}, where more details
can be found.

The planar subgraph $G_{\square}$ of $G_{\boxtimes}$    is
3-connected \cite{abk-sld3c-13} and has a rhomboidal canonical
ordering by assumption.  Graph $G_{\square}$ is processed according
to the canonical ordering using the shift technique as in
\cite{fpp-hdpgg-90} and extended to 3-connected graphs in
\cite{k-dpguco-96}. For every quadrilateral face $f = (a,b,c,d)$ in
 clockwise order with $a=bottom(f)$ and $c=top(f)$, the algorithm
first places $a$, then $b$ and $d$ in any order, and finally $d$
according to the canonical ordering. Vertex $b$ is placed on the
$-1$-diagonal through $a$ to the left and $d$ is placed on the
$+1$-diagonal through $a$ to the right and at the intersection with
the $+1$ and $-1$ diagonal of the left lower and right lower
neighbors, respectively. The technique in \cite{bdeklm-IC-16} is a
leveling of $b$ and $d$. If $b$ is placed $\delta$ units below $d$,
or vice versa, then  lift $b$  to the   level of $d$   by $2\delta$
extra shifts to the left. If $b$ has been leveled with other
vertices, then the shift is synchronously applied to all vertices
that are leveled with $b$. Alternatively, on may apply the critical
(or longest) path method so that the  critical paths to $b$ and $d$
have the same length.
At the placement of $c = top(f)$ we shift $b$ to the left or  $d$ to
the right so that $c$ is placed vertically above $d$. Then edge
$\{b,d\}$ is inserted as a horizontal line and $\{a,c\}$ as a
vertical one. Later on, $b,a$ and $d$ are shifted by the same amount
as $c$, so that the right angle crossing of $\{b,d\}$ and $\{a,c\}$
is preserved.
\end{proof}

There are rhomboidal 1-planar graphs that are not NIC-planar, such
as $k \times k$ grids with a pair of crossing edges in each inner
quadrangle and a triangulation of the outer face. (The graphs are
not NIC-planar, because they have too many edges). On the other
hand, every IC-planar graph admits a rhomboidal embedding and we
have  a simpler proof than in \cite{bdeklm-IC-16}.

\begin{corollary}
Every IC-planar graph is a RAC graph.
\end{corollary}

\begin{corollary}
There are 3-connected 1-planar graphs and NIC graphs that do not
admit a rhomboidal canonical ordering.
\end{corollary}
\begin{proof}
There are 1-planar graphs  \cite{del-dgrac-11} and even NIC-planar
graphs \cite{bbhnr-NIC-16} that are not RAC and a rhomboidal
canonical ordering would contradict Theorem~\ref{thm:rhomboidal}.
\end{proof}

We now turn to rectangle visibility representations of IC-planar
graphs and the proof of Theorem~\ref{thm:IC-RVG}.

A visibility representation of a 2-connected planar graph $G$ is
commonly obtained by  the following steps
\cite{dett-gdavg-99,k-mcvr-97}, which we call VISIBILITY-DRAWER:
\begin{enumerate}
  \item Compute an st-numbering $\delta(v)$ for the vertices of $G$ with
an edge $\{s,t\}$ and $\delta(s)=1$ and $\delta(t)=n$.  Embed edge
$\{s,t\}$  at the left and orient the edges according to the
st-numbering.
  \item Compute the $s^*t^*$-numbering of the dual graph $G^*$
where $s^*$ is the face to the right of $\{s,t\}$ and $t^*$ is the
outer face.
  \item For an
 oriented  edge $e$ let left$(e)$ (right$(e)$) be the $s^*t^*$-number
of the face to the left (right) of  $e$. For a vertex $v \neq s,t$
let left$(v) =$ min$\{$left$(e) \, | \, e$ is incident to $v\}$ and
right$(v) =$ max$\{$right$(e) \, | \, e$ is incident to $v\}$.
\item For each vertex $v \neq s,t$ draw a bar between $(\textrm{left}(v),
\delta(v))$ and $(\textrm{right}(v)-1, \delta(v))$ and draw a bar
between $(0,0)$ and $(M-1,0)$ for $s$ and between $(0,n-1)$  and
$(M-1,n-1)$ for $t$ where $M \leq 2n-4$ is the number of faces of
$G$.
\item Draw each edge $e= \{u,v\} \neq \{s,t\}$
between $(\textrm{left}(e), \delta(u))$ and $(\textrm{left}(e),
\delta(v))$ and draw $\{s,t\}$ at $x=1$.
\end{enumerate}

 There is exactly one vertex at each level $y = 1, \ldots, n$ if
 the st-numbers are used for the $y$-coordinates of the vertices. More compact drawings
are obtained by using the critical path method or topological
sorting \cite{dett-gdavg-99, k-mcvr-97}. The drawings are not really
pleasing, since many lines of sight are at the   ends of the bars.
There are no degenerated faces since the right end of a bar to the
left of   face $f$ is at least one unit to the left and of a bar to
the right of $f$. The drawing algorithm preserves the given
embedding.

  The change of the embedding at
B-configurations can be undone by modifying  the computation of the
$x$-coordinates of the lines of sight. Before the computation of the
dual $s^*t^*$-numbering add a copy of the rerouted $\{s,t\}$ edge at
its original place and remove the edge(s) that were added for the
3-connectivity and compute $\delta^*$ on the new embedding. The line
of sight for the edge $\{s,t\}$ can be drawn at several places
\cite{b-vr1pg-14} and we must choose the   one that preserves the
embedding.

\begin{algorithm}
  \caption{IC-RV-DRAWER}\label{alg:IC-RVG}

  \KwIn{An IC-planar embedding $\mathcal{E}(G)$.}

  \KwOut{A rectangle visibility representation $\mathcal{RV}(G)$.}

  Transform $\mathcal{E}(G)$ into a normal form embedding $\mathcal{E}(G_{\boxtimes})$.\;
  Compute the planar graph $G_{\square}$ and a rhomboidal st-numbering $\delta$ of $G_{\square}$. \;
   Compute   an $s^*t^*$-numbering $\delta^*$ of the dual graph  $G_{\square}^*$.   \;
  \ForEach{vertex v of $G_{\square}$} {
      \If{$v$  is the left (right) end   of a rhomboid and $u$ is the other end} {%
            $d(v) = \delta(u)+\delta(v)$
            }
      \Else { $d(v) = 2\delta(v)$}
     }
  Compute a planar visibility representation of $G_{\square}$ by \\ \quad VISIBILITY-DRAWER  with vertices on level $d(v)$
  and edges at left$(e)$.
    \ForEach{pair of crossing edges $\{a,c\}$ and $\{b,d\}$  \textrm{in a} \\
    \quad  \textrm{rhomboid}
    $f=(a,b,c,d)$ \textrm{with} $d(a) <  d(b)= d(d) < d(c)$ } {
   Add a horizontal line  of sight   at level $d(b)$  between the bars of  $b$
   and $d$.\;
   Add a vertical line line of sight at $\delta^*(f)+0.5$ between \\ \quad the bars of $a$ and
   $c$.
   }
   Scale all $x$-coordinates by two.\;
   Remove (or ignore) all lines of sight of edges not in $G$.
\end{algorithm}

The following Lemma concludes the proof of Theorem~\ref{thm:IC-RVG}.

\begin{lemma} \label{lem:correctnessALG1}
Algorithm IC-RV-DRAWER constructs a rectangle visibility
representation of an IC-planar graph on $O(n^2)$ area  and operates
in linear time.
\end{lemma}
\begin{proof}
The algorithm computes a planar visibility representation of
$G_{\square}$ in linear time as proved  in \cite{dett-gdavg-99,
rt-rplbopg-86, tt-vrpg-86} on an area of size   $(2n-5) \times 2n$,
which   is scaled by a  factor  of two in $x$-dimension.

Each pair of crossing edges  is in a kite of $G_{\boxtimes}$ whose
boundary is embedded as a rhombus $f=(a,b,c,d)$   with
$a=bottom(f)$. Then
 the y-coordinates of $b$ and $d$ coincide and $d$ is a weighted topological sorting as
 used in \cite{dett-gdavg-99}.
There is a gap of
 one unit  between the bars of $b$ and $d$, since the bar of $b$ ends at $\delta^*(f)-1$ and
 the bar of $d$ begins at $\delta^*(f)$. Now, edges $\{a,c\}$ and
$\{b,d\}$ are added so that they cross inside $f$.

Since $G$ has at most $13/4n-4$ edges, the transformation into
normal form takes linear time so that $G_{\boxtimes}$ and
$G_{\square}$ have size $O(n)$. The   visibility representation of
$G_{\square}$ is computed in linear time \cite{dett-gdavg-99}. There
are at most $n/4$ pairs of crossing edges which are each inserted in
$O(1)$ time.
\end{proof}

\section{T-Visibility of 1-planar Graphs} \label{sect:proof2}

For the \textsf{T}-visibility representation of 1-planar graphs we
use a leftish canonical ordering as an st-numbering and  draw
$G_{\square}$  by VISIBILITY-DRAWER. Note that $G_{\square}$ may
have multi-edges at separation pairs, which each introduces a face.
In total,  $G_{\square}$ has at most $2n-4$ faces, since each
multi-edge could be substituted by a planar edge. For the pairs of
crossing edges we expand some vertices to a $\bot$-shape. A
$\bot$-shaped vertex consists of a horizontal bar and a vertical
\emph{pylon}. By a horizontal flip we obtain a \textsf{T}-shape
visibility representation. If vertex $v$ is $\bot$-shaped, then the
 pylon  is inserted into the face of a left- or right-trapezoid $f$
with $v = bottom(f)$ and   $top(f)$ is maximum.   For each
quadrangle  $f= (v, b,c,d)$, the  edges $\{v,c\}$ and $\{b,d\}$ were
removed in the planarization step. They are reinserted as follows:
If $f$ is a rhomboid, then the lower of $b$ and $d$ gets a pylon for
an $\llcorner$- or $\lrcorner$-shape    so that $\{b,d\}$ is a
horizontal line that is crossed by
 a vertical line of sight for $\{a,c\}$ inside $f$.
If $f$ is a left-trapezoid, then  the bar of $b$ is extended to the
right and  edge $\{b,d\}$  is added as a vertical line of sight
inside $f$. Accordingly, extend the bar of $d$ to the left if $f$ is
a right-trapezoid.  The particularity is the drawing of edge
$\{v,c\}$ as a horizontal line of sight from the pylon of $v$ to the
bar (or pylon) of $c$, as depicted in Fig.~\ref{fig:rhomboids}. This
line of sight is unobstructed, since there is exactly one bar on
each level by the use of $st$-numbers for the $y$-coordinate  of the
vertices and there is no obstructing pylon from another vertex by
the use of the leftish canonical ordering and the bitonic order of
the vertices above a vertex $v$, as stated in Lemma
\ref{lem:bitonic}.

\begin{figure}
  \centering
  \subfigure[ ]{
    \parbox[b]{4.1cm}{%
      \centering
      \includegraphics[scale=0.7]{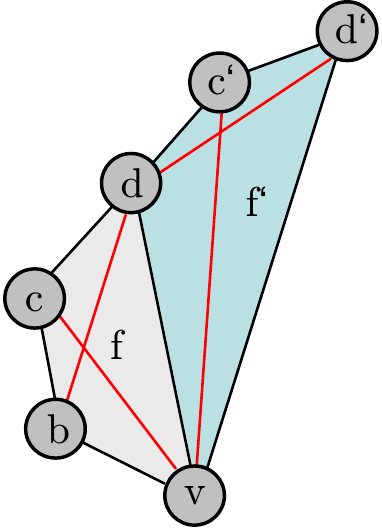}
    }
    \label{fig:kite1}
  }
  \hfil
  \subfigure[ ]{
    \parbox[b]{4.1cm}{%
      \centering
      \includegraphics[scale=0.9]{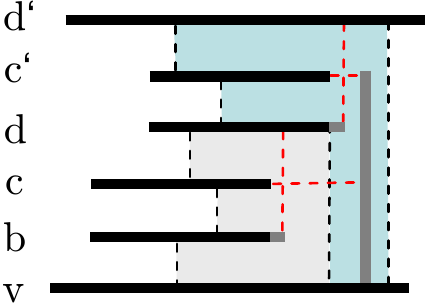}
    }
    \label{fig:B-config}
  }

  \caption{(a) Two left-trapezoids $f=(v, b,c,d)$ and
  $f'=(v,d,c',d')$ and (b) their visibility representation.}
  \label{fig:rhomboids}
\end{figure}

There is a special case at separation pairs and W-configurations,
see Fig.~\ref{fig:XWdouble}. If $G-\{x,y\}$ partitions into an outer
component $H_0$ and inner components $H_1,\ldots, H_p$, then
VISIBILITY-DRAWER places the inner components from left to right
between the bars of $x$ and $y$ and separates them in $x$-dimension
by the $st$-numbering and in $y$-dimension by the
$s^*t^*$-numbering. It  admits a representation of the copies of
edge $\{x,y\}$ by vertical lines of sight.

Consider the planarization $\widehat{H}_{\square}$ of an inner
component $H$ together with the separation pair $x,y$. The
 outer face of $\widehat{H}_{\square}$
 is a
quadrangle $f=(x, b,c,y)$ which is embedded as a left-trapezoid with
a copy of $\{x,y\}$ on the right. The vertex numbering from the
leftish canonical ordering is $x < b < w < c < y$, where $w \neq
b,c$ is any other vertex of $H_i$. Thus, $b$ and $c$ are the first
and last vertex of $H_i$. They are not unique, since the embedding
of $H_i$ can be flipped. In particular, $P_{q-1} = \langle c
\rangle$  and $P_q = \langle t \rangle$ are the last two paths in
the leftish canonical ordering of $\widehat{H}_{\square}$, since
there is no separating triangle $\Delta = (t,c,w)$ for some vertex
of $H$. Hence, $b$ and $c$ are placed on the lowest and highest
levels of the vertices of $H$. The outer edge  $\{b,t\}$ is
represented as a vertical line of sight between the bars of $b$ and
$t$ as if it were a planar edge, whereas    edge $\{x,c\}$ is a
horizontal line of sight from the pylon of $x$ to the pylon of $c$.
Algorithm T-DRAWER constructs the visibility representation.

\begin{figure}
   \begin{center}
     \includegraphics[scale=0.55]{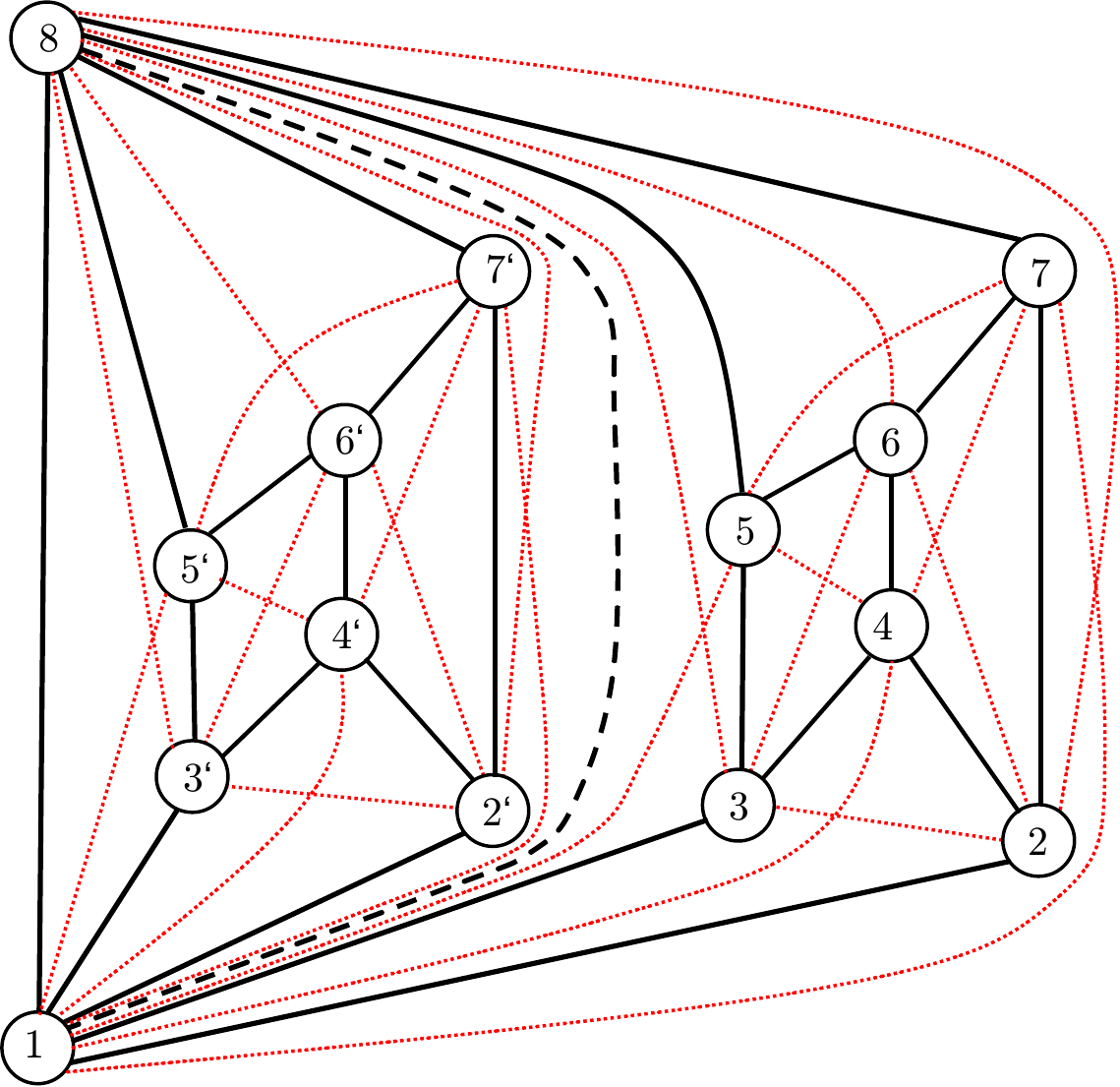}
  \caption{
A graph $DXW$ consisting of two copies of the extended wheel graph
$XW_6$ with   vertices $1, \ldots, 8$ and $1, 2',\ldots, 7', 8$.
Planar edges are drawn black and bold and crossing edges red and
dotted. The copy of edge $\{1,8\}$ is drawn dashed.
     The extended leftish canonical ordering is $\Pi = (\langle 1,2 \rangle, \langle 3,4 \rangle, \langle 5,6 \rangle,
     \langle 7 \rangle, \langle (1), 2' \rangle, \langle 3', 4' \rangle, \langle 5', 6' \rangle, \langle 7' \rangle,
     \langle 8 \rangle)$. }
     \label{fig:XWdouble}
   \end{center}
\end{figure}

\begin{algorithm}
  \caption{T-DRAWER}\label{alg:IC-RVG}

  \KwIn{A 1-planar embedding $\mathcal{E}(G)$.}

  \KwOut{A T-visibility representation $\mathcal{TVR}(G)$.}

  Compute  $\mathcal{E}(G_{\boxplus})$ from $\mathcal{E}(G)$.\;

  Compute $G_{\square}$ from $\mathcal{E}(G_{\boxplus})$ by removing all pairs of crossing edges.\;

  Compute an $st$-numbering $\delta$ of $G_{\square}$ as an extension
  of a leftish canonical \\ \quad ordering of each 3-connected component. \;
  Compute the $s^*t^*$-numbering $\delta^*$ of the dual graph
  $G_{\square}^*$. \;

  Compute the planar visibility representation of $G_{\square}$ by \\ \quad VISIBILITY-DRAWER.\;

  \ForEach{vertex $v$ with quadrangles $f$ such that $v =bottom(f)$ } {
    \textbf{int} $v_{max} = 0$;  \textbf{face} $f_{max}$ \;
      \ForEach{left-trapezoid $f=(v,b,c,d)$} {
        \If{$v_{max} < c$} {
         $v_{max} =   c$; \,
         $f_{max}=f$
         }
        Extend the bar of $b$ by $1/3$  to the right and at its right end
        \;
        add a vertical line of sight for $\{b,d\}$.
  }
 \ForEach{right-trapezoid $f=(v,d,c,b)$} {
      \If{$v_{max} < c$} {
         $v_{max} = c$; \,
         $f_{max}=f$
         }
        Extend the bar of $b$ by $1/3$  to the left and  at its left end \;
        add a vertical line of sight for $\{b,d\}$.
   }
  \ForEach{rhombus $f=(v,b,c,d)$} {
       Enlarge the bar of the lower of $b$ and $d$ by a pylon  at an end   \\ \quad
        and inside $f$ and up to the
       bar of the upper vertex. \;
       Add a horizontal  line of sight for $\{b,d\}$  at the top of the pylon.\;
      Add a vertical line of sight for $\{v,c\}$ between the bars of
      $v$ and $c$ \\ \quad at $\delta^*(f)+1/3$\;
    }

  \If{$v_{max} \neq 0$} {
       Enlarge the bar of $v$ by a pylon inside $f_{max}(v)$
       from \\ \quad $(\delta^*(f_{max}(v))+1/3, \delta(v))$ to  $(\delta^*(f_{max}(v))+1/3 , \delta(v_{\max}))$\;
      \ForEach {trapezoid  $f=(v,b,c,d)$ } {%
          \If {$f$ is not a left-trapezoid with  a separation pair  $[v,d]$ }
          {
            Add a horizontal  line of sight for $\{v,c\}$ \\ \quad from the
           pylon  to the bar of $c$.
      }
      \Else (\algcom{the horizontal line of sight may be occupied}) {
      Enlarge the bar of $c$ by a pylon of height $1/2$.  \;
      Add a horizontal  line of sight for $\{v,c\}$  at the top of the pylon of $c$.
} } } }
Scale all $x$-coordinates by three and all y-coordinates by two.\;
\end{algorithm}

 \pagebreak

\begin{lemma} \label{lem:pylonSee}
Suppose   $G$ is a 3-connected 1-planar graph and there is no
W-configuration in the outer face of an embedding of $G$. For a
vertex $v$ on a contour $C_k$, let $f_{j_1}(v),\ldots,
f_{j_{\mu}(v)}$ be the sequence of left- and right-trapezoids above
$v$ from left to right with $f_i(v) = (v, b_i, c_i, d_i)$. Let
$v_{max} = \max \{ c_i \, | \, i=j_1,\ldots, j_{\mu}\}$ and let
$f_{max}(v)$ be the trapezoid containing $v_{max}$.

Then the  pylon  of $v$  inside   $f_{max}(v)$
 can see the bar  of each vertex $c_i$ for $i \in \{j_1,\ldots, j_{\mu}\}$.
\end{lemma}

\begin{proof}
By Lemma \ref{lem:bitonic}, there is a bitonic sequence of clockwise
neighbors of $v$, which each has is own $y$-coordinate
 according to the leftish canonical ordering. Hence, a horizontal
 line of sight from the pylon of $v$ is
unobstructed by bars of other vertices. A horizontal line  of sight
from the pylon to $c_i$ intersects only vertical lines of sight of
planar edges $\{v, w\}$ with $c_i < w$ and $c_i$ and $w$ are on  the
same side of the pylon, i.e., $c_i, w < w_m$ or $c_i, w > w_m$,
where $w_m$ is the maximum neighbor of $v$ (or the top vertex of
$f_{j_{\mu}}(v) $) in the leftish canonical ordering. Hence, a line
of sight $\{v, c_i\}$ is unobstructed by pylons of other vertices.
In consequence, each edge $\{v, c_i\}$ with $i \in \{j_1,\ldots,
j_{\mu}\}$ is represented in the visibility representation
constructed by T-DRAWER.
\end{proof}

Finally, consider a separation pair $[x,y]$ with inner components
$H_1,\ldots, H_p$.
The st-numbering extending the leftish canonical ordering of
3-connected components inserts the vertices of each component
consecutively and just before $y$ so that there is a subsequence $x,
H_0', H_1, \ldots, H_p, y$, where $H_0'$ is a subgraph of the outer
component that is added by the leftish canonical ordering   between
  $x$ and $y$. Each component $H_i$ is drawn in a box
$B(H_i)$ and the boxes are ordered monotonically in $x$- and in
$y$-dimension to a staircase between the bars of $x$ and $y$ both by
the common visibility drawer and by  T-DRAWER, as illustrated in
Fig.~\ref{fig:XWplanarVis}.

Consider the outer face of an inner component including the
separation pair. Without crossing edges, there is a quadrangle
$f_{out}(H) = (x, b,c,y)$, which is embedded as a left-trapezoid.
However, $f_{out}(H)$ has no left-support, since $x<b<c < y$ in the
leftish canonical ordering and edge $\{x,y\}$ of $f_{out}(H)$ is a
copy of the original edge. This case is treated as an
\emph{exception}. Edge $\{b,y\}$ is drawn inside $f_{out}(H)$ and to
the right of $H$ after an extension of the bar of $b$ to the right.
Vertex $x$ is $\bot$-shaped with a high pylon up to $y$ which is
placed in the face to the right of the original edge $\{x,y\}$.
 The pylon can see all vertices that
are neighbors of $x$ in the trapezoids of $H$  by Lemma
\ref{lem:pylonSee}. However, the horizontal line of sight to $c$ may
be occupied, as in Fig.~\ref{fig:XWTVis}. Fortunately, $c$ is the
last vertex of $H$ in the leftish canonical ordering and a short
pylon for the bar of $c$ admits a horizontal   line of sight between
$x$ and $c$. Since the inner components  are separated in
$y$-dimension, the pylon of $x$ can see all neighbor of $x$ in the
trapezoids of the inner components.

 The following Lemma concludes the proof of Theorem \ref{thm:1-cat}.

\begin{lemma} \label{lem:correctALG2}
Algorithm T-DRAWER constructs a \textsf{T}-visibility representation
of a 1-planar graph on $O(n^2)$ area  and operates in linear time.
\end{lemma}
\begin{proof}
The computations  of   $G_{\boxplus}$, the removal
 of all pairs of crossing edges for $G_{\square}$, the st-numbering as an extension of a leftish canonical
 ordering, the $s^*t^*$-numbering and the planar visibility
 representation of $G_{\square}$ each take linear time if a 1-planar
 embedding of $G$ is given.
There are
 at most n-2 pairs of crossing edges which can each be inserted in
 $O(1)$ time into the visibility
 representation of $G_{\square}$. Hence, T-DRAWER runs in liner
 time. The visibility
 representation of $G_{\square}$ has size at most $(2n-5) \times n$, which is expanded by a factor of six.

 The common visibility drawer provides a correct visibility
 representation of $G_{\square}$. For each 3-connected component
 without a W-configuration,
 the pairs of crossing edges are correctly added to the visibility
 representation by Lemma \ref{lem:pylonSee}. The pair of
 edges crossing in the outer face of a W-configuration is visible by
 the special treatment in lines 34 and 35. Since inner components at a
 separation pair $[x,y]$ are strictly separated in both dimensions
 and are placed between the bars (shapes) of $x$ and $y$, there is a line of sight
 between the shapes of $x$ and $y$ for each edge
 between $x$, $y$ and vertices of inner components.
 Finally, consider the decomposition tree.
If $H$ is an inner component at a separation pair $[x,y]$, then
there is no edge $\{u,v\}$ from a vertex $u$ with $x \neq u \neq y$
of the outer component to a vertex $v$ of $H$  and, hence, there is
no need for a line of sight. In addition, there is no need for a
horizontal line of sight from a pylon through the visibility
representation of an inner component, since the st-numbering groups
components recursively and thereby separates them. Hence, the pylons
in inner components do not obstruct horizontal lines of sight from
pylons of vertices of the outer component.
\end{proof}

It is important to use weak visibility, since a pylon can see the
bars and pylons of many other vertices, which is forbidden in the
strong visibility version.

As an example, consider the extended wheel graph $XW_6$
\cite{s-s1pg-86} and then take copies of it and identify two
vertices, here 1 and 8. These graphs have been used for the
construction of sparse maximal 1-planar graphs
\cite{begghr-odm1p-13} and for a linear lower bound on the number of
legs (vertex complexity) in embedding-preserving caterpillar-shape
visibility representations \cite{ddelmmw-ovreg-16}. Graph $XW_6$ can
be seen as a cube in 3D in which each face contains a pair of
crossing edges.

The visibility representation of $G_{\square}$ from the common
visibility drawer is displayed in Fig.~\ref{fig:XWplanarVis} and the
$\bot$-shape visibility representation of  T-DRAWER  in
Fig.~\ref{fig:XWTVis}. Note that the graphs even admit a rectangle
visibility representation (use the high pylons of vertices $1, 2,
5,2',5'$ and fill $4$ and $4'$ to a rectangle  in
Fig.~\ref{fig:XWTVis}).

\begin{figure}
   \begin{center}
     \includegraphics[scale=0.5]{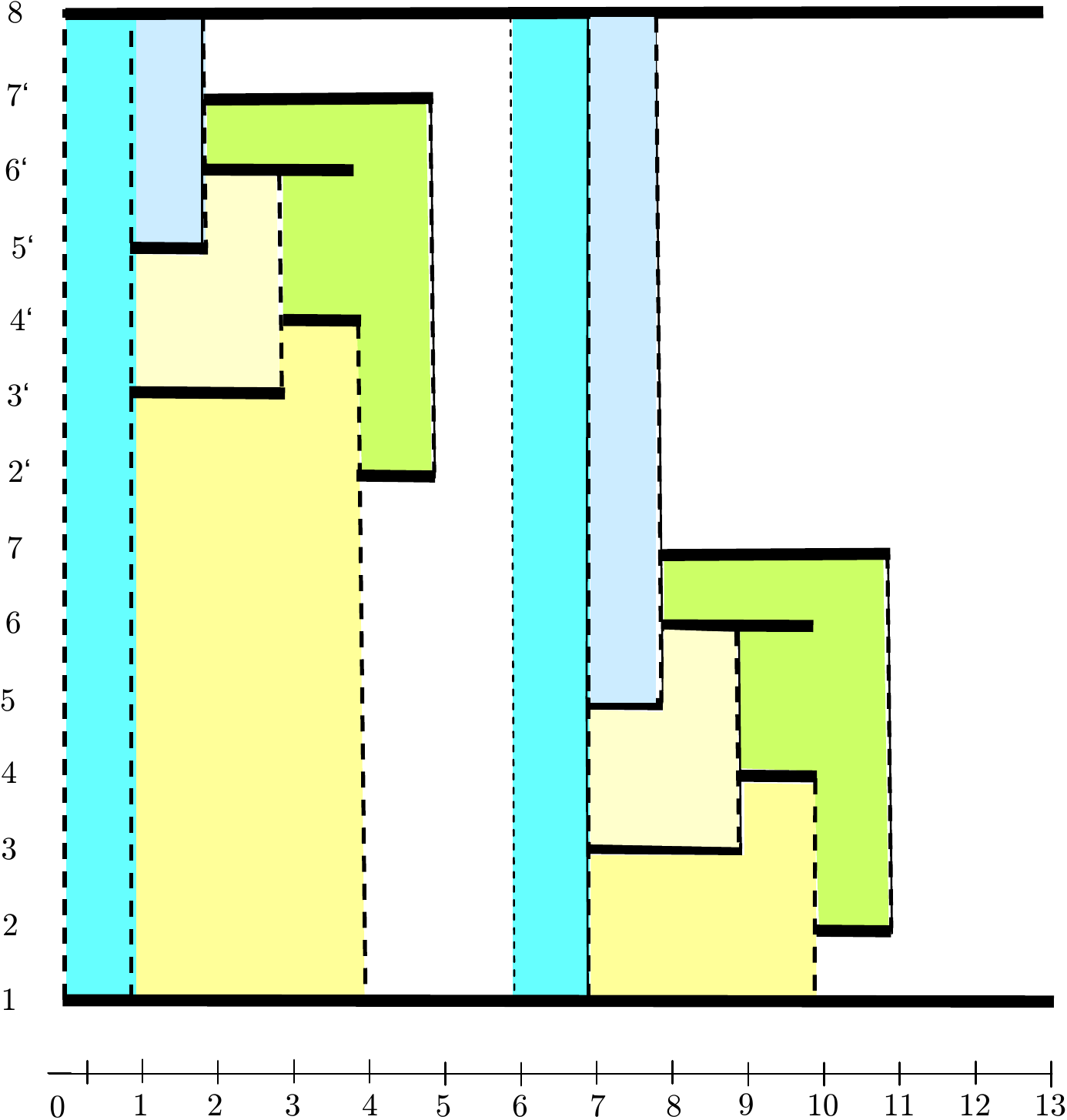}
     \caption{A visibility representation of $DXW_{\square}$ from Fig.~\ref{fig:XWdouble}
     with dashed lines of sight and colored faces.}
     \label{fig:XWplanarVis}
   \end{center}
\end{figure}

\begin{figure}
   \begin{center}
     \includegraphics[scale=0.5]{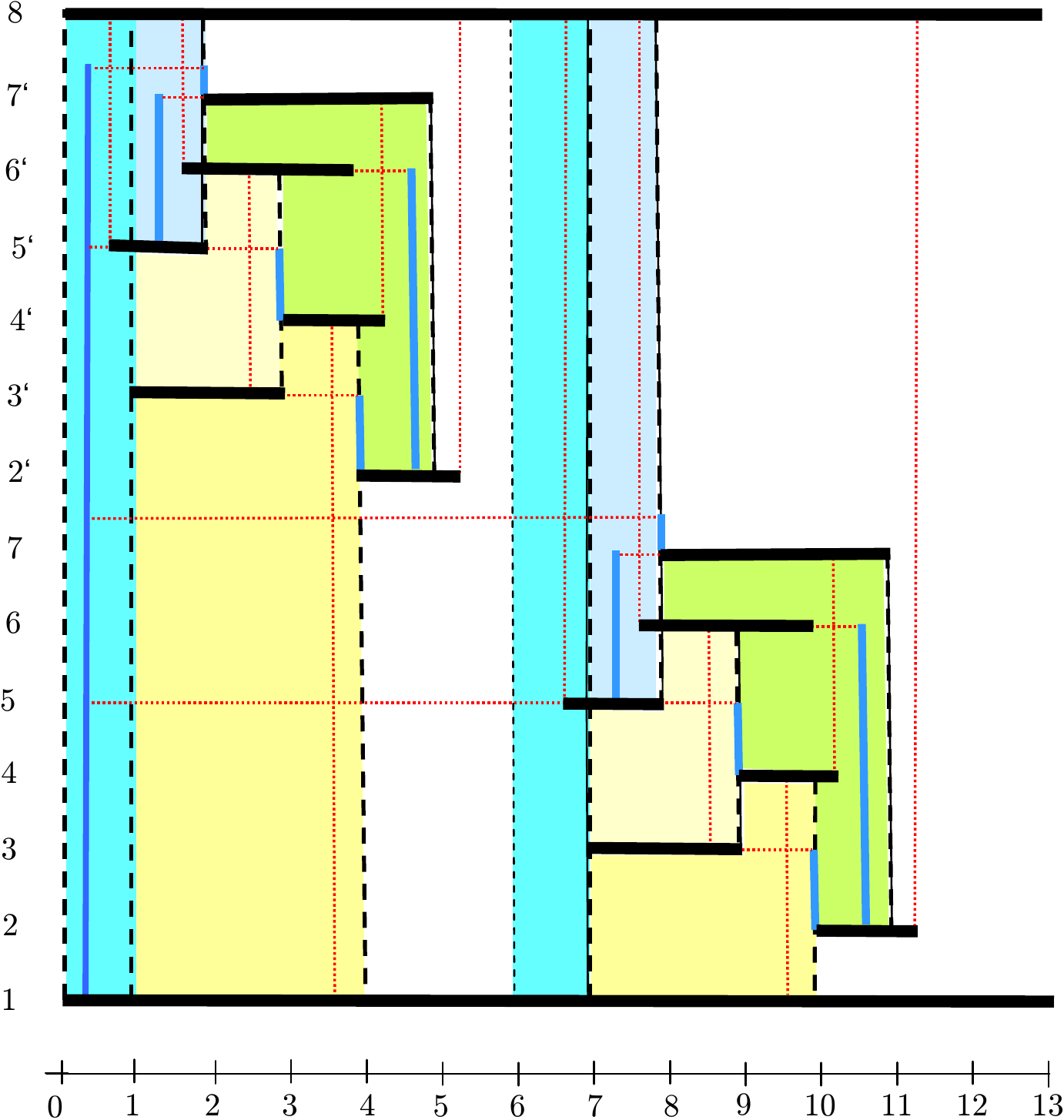}
     \caption{The \textsf{T}-shaped visibility representation of graph $DXW$ from Fig.~\ref{fig:XWdouble}
     by T-DRAWER (with pylons in blue).}
     \label{fig:XWTVis}
   \end{center}
\end{figure}

\iffalse
\begin{figure}
   \begin{center}
     \includegraphics[scale=0.7]{XW6rect}
     \caption{OPTIONAL:  A rectangle  visibility representation of $XW_6$.}
     \label{fig:XWTVis}
   \end{center}
\end{figure}
\fi

\section{General Shape Visibility Graphs} \label{sect:general}

There is a natural ordering relation $\sigma < \sigma'$ between
shapes     if  $\sigma$ is a restriction of $\sigma'$ including
rotation and flip. For example, $\textsf{I} < \textsf{L} <
\textsf{F} < \textsf{E}$ and $\textsf{I}< \textsf{T} < \textsf{E} <
rake < caterpillar$. Clearly, every $\sigma$-shape visibility graph
is a $\sigma'$-shape visibility graph if   $\sigma < \sigma'$.
However, it is unclear whether different shapes imply different
classes of shape visibility graphs.
 Moreover, shapes with
cycles, such as \textsf{O} or \textsf{B} are not really useful for
shape visibility representations, since a cycle corresponds to an
articulation vertex.

For shape visibility graphs we can  state:

\begin{lemma} \label{lem:thickness}
Every   shape visibility graph  has thickness two.
\end{lemma}

\begin{proof}
The subgraph induced by the horizontal (vertical) lines of sight is
planar.
\end{proof}

\begin{corollary}
  $\sigma$-visibility graphs of size $n$ have at
most $6n-12$ edges and there are $\sigma$-visibility graphs with
$6n-20$ edges for every shape $\sigma$.
\end{corollary}

The upper bound follows from Lemma \ref{lem:thickness} and the lower
bound has been proved by Hutchinson et al.~\cite{hsv-rstg-99} for
rectangle visibility graphs. The exact  bound are unclear for all
shapes except rectangles.\\

The extended wheel graph $XW_6$ even admits a rectangle visibility
representation, and so do all wheel graphs $XW_{2k}$ with $k \geq
3$. An extended wheel graph consists of a cycle of vertices $v_1,
\ldots, v_{2k}$ of vertices of degree six so that each $v_i$ is
adjacent to its next and next but one vertex in cyclic order. In
addition, there are two poles $p$ and $q$ that are adjacent to all
$v_i$ (but there is no edge $\{p,q\})$. Extended wheel graphs  play
a prominent role for 1-planar graphs  with $4n-8$
edges \cite{b-ro1plt-16, s-s1pg-86, s-rm1pg-10}.\\

We close with some open problems:\\

\noindent\textbf{ Conjecture}:
\begin{enumerate}
  \item Every   1-planar graph with $4n-8$ edges is a rectangle visibility
  graph.
  \item There are \textsf{L}-visibility graphs that are not rectangle
visibility graphs (\textsf{I}-shape) and there are
\textsf{T}-visibility graphs that are not \textsf{L}-visibility
graphs.
\end{enumerate}

\section{Acknowledgement}
I wish to thank Christian Bachmaier for his useful comments and
suggestions.

\bibliographystyle{abbrv}
\bibliography{brandybibV8}

\end{document}